\documentclass[openacc]{rstransa}

\usepackage{stmaryrd}
\usepackage{tikz}
\usepackage{tikz-cd}
\usepackage{mathtools}
\usepackage[numbers]{natbib}
\usepackage{hyperref}
\usepackage[mathscr]{euscript}

\newcommand{\SG}{\mathsf{SG}}
\newcommand{\SSet}{\mathsf{SS}}

\newcommand{\XORT}{\mathbb{T}^{\oplus}}
\newcommand{\XORTe}[1]{\XORT_{#1}}
\newcommand{\Theories}{\oplus\text{-}\textbf{Th}}
\newcommand{\Structures}{\oplus\text{-}\textbf{Str}}
\newcommand{\XORM}{\mathbb{M}^\oplus}
\newcommand{\XORMt}[1]{\XORM_{#1}}

\newcommand{\tuple}[1]{\left\langle #1 \right\rangle}

\newcommand*{\LargerCdot}{\raisebox{-0.8ex}{\scalebox{2}{$\cdot$}}}
\newcommand{\action}{\LargerCdot}

\newcommand{\GHZt}{\mathsf{GHZ}}
\newcommand{\GHZ}{\ket{\GHZt}}

\newcommand{\GHZnt}{\mathsf{GHZ}(n)}
\newcommand{\GHZn}{\ket{\GHZnt}}
\newcommand{\comment}[1]{}

\newcommand{\ket}[1]{|#1\rangle}
\newcommand{\braket}[2]{\langle#1|#2\rangle}

\newcommand{\kp}{\ket{\psi}}

\newcommand{\Pauli}{\mathsf{P}}
\newcommand{\PN}{\Pauli_n}

\newcommand{\Clifford}{\mathsf{C}}

\newcommand{\rank}{\text{rank}~}

\newcommand{\biff}{\quad\Leftrightarrow\quad}

\newcommand{\U}{\mathcal{U}}
\newcommand{\MX}{\mathscr{X}}

\newtheorem{theorem}{\bf Theorem}[section]

\newtheorem{corollary}{\bf Corollary}[section]
\newtheorem{proposition}{\bf Proposition}[section]
\newtheorem{lemma}{\bf Lemma}[section]

\theoremstyle{definition}
\newtheorem{defn}{\bf Definition}[section]

\newcommand{\ie}{\textit{i.e.}~}

\jname{rsta}
\Journal{Phil. Trans. R. Soc. A}

\begin{document}

\title{A complete characterisation of All-versus-Nothing arguments for stabiliser states}

\author{
Samson Abramsky$^{1}$, Rui Soares Barbosa$^{1}$, Giovanni Car\`{u}$^{1}$, Simon Perdrix$^{2}$}

\address{$^{1}$Department of Computer Science, University of Oxford, Wolfson Building, Parks Road, Oxford OX1 3QD, U.K.\\
$^{2}$ CNRS, LORIA, Universit\'{e} de Lorraine, Campus Scientifique, F54506 Vand\oe uvre-l\`es-Nancy, France}

\subject{Quantum information, quantum foundations}


\corres{
\email{samson.abramsky@cs.ox.ac.uk}\\
\email{rui.soares.barbosa@cs.ox.ac.uk}\\
\email{giovanni.caru@cs.ox.ac.uk}\\
\email{simon.perdrix@loria.fr}}

\begin{abstract}
An important class of contextuality arguments in quantum foundations are the All-versus-Nothing (AvN) proofs, generalising a construction originally due to Mermin.
We present a general formulation of All-versus-Nothing arguments, and a complete characterisation of all such arguments which arise from stabiliser states.
We show that every AvN argument for an $n$-qubit stabiliser state can be reduced to an AvN proof for a three-qubit state which is local Clifford-equivalent to the tripartite GHZ state.
This is achieved through a combinatorial characterisation of AvN arguments, the AvN triple Theorem, whose proof makes use of the theory of graph states.
This result enables the development of a computational method to generate all the AvN arguments in $\mathbb{Z}_2$ on  $n$-qubit stabiliser states.
We also present new insights into the stabiliser formalism and its connections with logic.
\end{abstract}


\begin{fmtext}
\section{Introduction}
Since the  classic no-go theorems by Bell \cite{Bell} and Bell--Kochen--Specker \cite{Kochen, Bell-KS}, contextuality has gained great importance in the development of quantum information and computation. This key characteristic feature of quantum mechanics represents one of the most valuable resources at our disposal to break through the limits of classical computation and information processing, with various concrete applications e.g.
in quantum computation speed-up \cite{Howard, Raussendorf4}
and in device-independent quantum security \cite{Horodecki}. Bell's proof of his famous theorem, as well as the subsequent formulation due to Clauser, Horne, Shimony, and Holt~\cite{CHSH}, relies on the derivation of
\end{fmtext}
%
%
\maketitle
\noindent
an inequality involving empirical probabilities, which must be satisfied by any local theory but is violated by quantum mechanics.
Greenberger, Horne, Shimony, and Zeilinger~\cite{Greenberger, Greenberger2} gave a stronger proof of contextuality, which only revolves around possibilistic features of quantum predictions, without taking into account the actual values of the probabilities. This version of the phenomenon corresponds to the notion of \emph{strong contextuality} in the hierarchy introduced by \cite{Abramsky1}.
In 1990, Mermin presented a simpler proof of this phenomenon, which rests on deriving an inconsistent system of equations in $\mathbb{Z}_2$, and became known in the literature as the original  ``All-vs-Nothing'' (AvN) argument~\cite{Mermin}.
Recent work on the mathematical structure of contextuality \cite{Abramsky1} allowed a powerful formalisation and generalisation of Mermin's proof to a large class of examples in quantum mechanics using stabiliser theory \cite{Abramsky2}.
This work is of particular significance in the context of measurement-based quantum computation (MBQC), where AvN contextuality plays an important r\^ole \cite{Hoban, Raussendorf4,Okay17}.
The general theory of AvN arguments introduced in \cite{Abramsky2} raises the natural question of whether it is possible to identify the states admitting this type of proof of contextuality. In the present paper, we address this question by combining the general theory of AvN arguments with the graph state formalism~\cite{Hein}. We summarise our results:
\begin{itemize}
\item We show that generalised AvN arguments on stabiliser states can be completely characterised by \emph{AvN triples}, proving part of what was previously known as the \emph{AvN triple conjecture} \cite{AvN}. 
\end{itemize}
AvN triples are triples of elements of the Pauli $n$-group satisfying certain combinatorial properties.
The presence of such a triple in a stabiliser group  was proved in \cite{Abramsky2} to be a sufficient condition for AvN contextuality.
We show that the converse of this statement is also true for the case of \emph{maximal} stabiliser subgroups, or equivalently, for stabiliser states, yielding a complete combinatorial characterisation of AvN arguments using stabiliser states. 
This new description of AvN arguments leads to another important result:
\begin{itemize}
\item We prove that any AvN argument on an $n$-qubit stabiliser state can be reduced to an AvN argument only involving three qubits which are local Clifford-equivalent to the tripartite GHZ state. 
\end{itemize}

The last part of the paper is dedicated to an application of this characterisation:
\begin{itemize}
\item We present a computational method to generate all AvN arguments for $n$-qubit stabiliser states (for $n$ sufficiently small).
\end{itemize}
Until now, we had a rather limited number of examples of quantum-realisable strongly contextual models giving rise to AvN arguments. The technique we introduce here gives us a large amount of instances of this type of models.

These findings are introduced after a theoretical digression on general All-vs-Nothing arguments for stabiliser states, which provides interesting new insights on their relationship to logic. In particular:
\begin{itemize}
\item We represent the well-known link between subgroups of the Pauli $n$-group and their stabilisers in the Hilbert space of $n$-qubits as a Galois connection, which allows us to establish a relationship with the Galois connection between syntax and semantics in logic (see e.g. \cite{Smith})
\end{itemize}
Previous work has already shown intriguing connections between logic and contextuality \cite{Abramsky4,Abramsky2,Kishida16:logic,AbramskyEtAl17:QuantumMonad}.
For instance, a direct link between the structure of quantum contextuality and classic semantic paradoxes is shown in \cite{Abramsky2}.
The Galois correspondence we exhibit is a step towards a precise understanding of this interaction.

\paragraph{Outline} The remainder of the paper is organised as follows. In Section \ref{sec: Mermin}, we recall the original All-vs-Nothing argument by Mermin, and generalise it to a class of  empirical models. Section \ref{sec: stabiliser} introduces the stabiliser formalism, the Galois correspondence, and the connections with logic. In Section \ref{sec: AvN triple conjecture}, we prove the AvN triple theorem and investigate its consequences. Finally, in Section \ref{sec: method}, we present and illustrate the method to generate AvN triples. 

\section{All-vs-Nothing arguments}\label{sec: Mermin}

\subsection{Mermin's original formulation}
In this section, we review Mermin's AvN proof of strong contextuality of the GHZ state.

Recall the definition of the \textbf{Pauli operators}, dichotomic observables corresponding to measuring spin in the $x$, $y$, and $z$ axes, with eigenvalues $\pm 1$
\[
X\coloneqq \left(
\begin{matrix}
0 & 1 \\
1 & 0
\end{matrix}
\right)
~~
Y\coloneqq \left(
\begin{matrix}
0 & -i \\
i & 0
\end{matrix}
\right)
~~
Z\coloneqq \left(
\begin{matrix}
1 & 0 \\
0 & -1
\end{matrix}
\right)
\]
These matrices are self-adjoint, have eigenvalues $\pm 1$, and together with identity matrix $I$ satisfy the following relations:
\begin{gather}
X^2=Y^2=Z^2=I \nonumber \\ 
XY=iZ,~~ YZ=iX,~~ ZX=iY, \label{equ: Pauli's}\\
YX=-iZ,~~ZY=-iX,~~XZ=-iY. \nonumber
\end{gather}

The GHZ state is a tripartite  qubit state, defined by
\[
\GHZ\coloneqq \frac{1}{\sqrt{2}}(\ket{000}+\ket{111}).
\]
We consider a tripartite measurement scenario where each party $i=1,2,3$ can perform a Pauli measurement in $\{X_i, Y_i\}$ on GHZ,
obtaining, as a result, an eigenvalue in $\{\pm1\}$.\footnote{It is convenient to relabel $+1, -1, \times$ as $0,1,\oplus$ respectively, where $\oplus$ is addition modulo $2$. The eigenvalues of a joint measurement $A_1\otimes A_2\otimes\dots\otimes A_n$ are the products of eigenvalues at each site, so they are also $\pm1$ Thus, joint measurements are still dichotomic and only distinguish joint outcomes up to parity\label{note1}.} By adopting the viewpoint of \cite{Abramsky1}, this experiment can be seen as an empirical model whose support is partially described by Table \ref{tab: GHZ}
(the remaining choices of measurements give rows with full support and can therefore be ignored).
\begin{table}[htbp]
\centerline{
\begin{tabular}{c c c | c c c c c c c c}
\hline
$1$ & $2$ & $3$ & $000$ & $001$ & $010$ & $011$ & $100$ & $101$ & $110$ & $111$\\
\hline
\dots & \dots & \dots & \dots & \dots & \dots & \dots & \dots & \dots & \dots & \dots\\
$X_1$ & $X_2$ & $X_3$ & $1$ & $0$ & $0$ & $1$ & $0$ & $1$ & $1$ & $0$ \\
$X_1$ & $Y_2$ & $Y_3$ & $0$ &  $1$ &  $1$ &  $0$ &  $1$ &  $0$ &  $0$&  $1$\\
$Y_1$ & $X_2$ & $Y_3$ & $0$ &  $1$ &  $1$ &  $0$ &  $1$ &  $0$ &  $0$&  $1$\\
$Y_1$ & $Y_2$ & $X_3$ & $0$ &  $1$ &  $1$ &  $0$ &  $1$ &  $0$ &  $0$&  $1$\\
\dots & \dots & \dots & \dots & \dots & \dots & \dots & \dots & \dots & \dots & \dots\\
\end{tabular}
}
\caption{A partial table for the support of the GHZ model.}\label{tab: GHZ}
\end{table}

\noindent These partial entries of the table can be characterised by the following equations in $\mathbb{Z}_2$:
\begin{alignat*}{14}
\bar{X}_1 & {} \oplus {} & \bar{X}_2 & {} \oplus {} & \bar{X}_3 & {} = 0&
&\qquad\qquad\qquad\qquad&
\bar{Y}_1 & {} \oplus {} & \bar{X}_2 & {} \oplus {} & \bar{Y}_3 & {} = 1&
\\
\bar{X}_1 & {} \oplus {} & \bar{Y}_2 & {} \oplus {} & \bar{Y}_3 & {} = 1&
&&
\bar{Y}_1 & {} \oplus {} & \bar{Y}_2 & {} \oplus {} & \bar{X}_3 & {} = 1&,
\end{alignat*}
where $\bar{P}_i\in\mathbb{Z}_2$ denotes the outcome of the measurement $P_i$, for all $P_i\in\{X_i, Y_i\}$. 
It is straightforward to see that this system is inconsistent. Indeed, if we sum all the equations, we obtain $0=1$, as each variable appears twice on the left-hand side. This means that we cannot find a global assignment $\{X_1, Y_1, X_2, Y_2, X_3, Y_3\}\rightarrow\{0,1\}$ consistent with the model, showing that the GHZ state is strongly contextual. 

\subsection{General setting}
The description of empirical models as generalised probability tables provided by \cite{Abramsky1} allows us to generalise Mermin's argument to a larger class of examples. 

A \textbf{measurement scenario} is a tuple $\tuple{\MX,\U,O}$ where
$\MX$ is a finite set of measurement labels, $O$ a finite set of measurement outcomes,
and $\U$ is
a family $\U \subseteq \mathcal{P}(\MX)$ that is downwards-closed and satisfies $\bigcup \U=\MX$.
The elements of $\U$ are the \textbf{measurement contexts},
i.e. the sets of measurements that can be  performed jointly.
An \textbf{empirical model} over the scenario $\tuple{\MX,\U,O}$ is a family $\{e_U\}_{U\in\U}$ indexed by the contexts, where $e_U$ is a probability distribution on $O^U$, the set of joint outcomes for the measurements in the context $U$.
We can think of an empirical model as a table with rows indexed by the contexts $U \in \U$, columns indexed by the joint outcomes $s \in O^U$, and entries given by the probabilities $e_U(s)$.
In the examples considered here, we typically only present the rows corresponding to maximal contexts, as probabilities of outcomes for smaller contexts are determined from these by marginalisation -- this corresponds to a compatibility condition which is a generalised form of no-signalling \cite{Abramsky1}.


An empirical model $e$ is \textbf{strongly contextual} if there is no global assignment $g :\MX \rightarrow O$ which is \emph{consistent} with $e$ in the sense that the restriction of $g$ to each context $U$ yields a joint outcome for that context which is possible (has non-zero probability) in $e$.
Formally, there is no $g : \MX \rightarrow O$ such that, for all $U \in \U$, $e_U(g |_U) >0$.

In the setting we will consider, all the measurements are dichotomic and produce outcomes in $\mathbb{Z}_2$.\footnote{More generally, one can define the notion of AvN arguments for models with outcomes valued in a unital, commutative ring. However, for the purposes of this paper, we are only concerned with All-vs-Nothing parity arguments arising from stabiliser quantum theory.} 

To an empirical model $e\coloneqq \{e_U\}_{U\in\U}$ we associate an \textbf{XOR theory} $\XORTe{e}$.
It uses a $\mathbb{Z}_2$-valued variable $\bar{x}$ for each measurement label $x \in \MX$.
For each context $U\in\U$, $\XORTe{e}$ will include the assertion 
\[
\bigoplus_{x\in U} \bar{x}=0,
\]
whenever the support of $e_U$ only contains joint outcomes of even parity
i.e. with an even number of $1$s,
and 
\[
\bigoplus_{x\in U} \bar{x}=1
\]
whenever it only contains joint outcomes of odd parity, i.e. with an odd number of $1$s.
In other words, it includes the assertion $\bigoplus_{x \in U} \bar{x} = o$ whenever $\bigoplus_{x \in U}s(x) = o$ for all $s \in \mathbb{Z}_2^U$ such that $e_U(s) > 0$.

We say that the model $e$ is \textbf{AvN} if the theory $\XORTe{e}$ is inconsistent. Since an inconsistent theory implies the impossibility of defining a consistent global assignment $\MX\rightarrow \mathbb{Z}_2$, we have the following result.

\begin{proposition}[{\cite[Proposition 7]{Abramsky2}}]
If an empirical model $e$ is AvN, then it is strongly contextual. 
\end{proposition}
As an example, we consider the Popescu--Rohrlich (PR) Box model \cite{PR} given in Table \ref{tab: PR}. 
\begin{table}[htbp]
\centering
\begin{tabular}{c c | c c c c}
\hline
$A$ & $B$ & $00$ & $10$ & $01$ & $11$\\
\hline
$a_1$ & $b_1$ & $1$ & $0$ & $0$ & $1$\\
$a_1$ & $b_2$ & $1$ & $0$ & $0$ & $1$\\
$a_2$ & $b_1$ & $1$ & $0$ & $0$ & $1$\\
$a_2$ & $b_2$ & $0$ & $1$ & $1$ & $0$\\
\end{tabular}
\caption{the PR-box model}\label{tab: PR}
\end{table}

\noindent The XOR theory of the PR box model consists of the following 4 equations
\[
\arraycolsep=1pt
\begin{array}{ccccc}
\bar{a}_1 & \oplus & \bar{b}_1 & = & 0\\
\bar{a}_1 & \oplus & \bar{b}_2 & = & 0
\end{array}\qquad\qquad\qquad
\begin{array}{cccccc}
\bar{a}_2 & \oplus & \bar{b}_1 & = & 0\\
\bar{a}_2 & \oplus & \bar{b}_2 & = & 1,
\end{array}
\]
which lead to $0=1$. Hence, the theory is inconsistent and the model is strongly contextual.

Although it is obviously not always the case that parity equations can fully describe the support of an empirical model, this setting is particularly suited for the study of strong contextuality in stabiliser states, as we shall see in the following sections. 

\section{The stabiliser world}\label{sec: stabiliser}

Stabiliser quantum mechanics \cite{Chuang, Caves} is a natural setting for general All-vs-Nothing arguments, and allows us to see how AvN models can arise from quantum theory. In this section, we recall the main definitions and present new insights on the connection between AvN arguments and logical paradoxes. 

\subsection{Stabiliser subgroups}

Let $n\ge 1$ be an integer. The \textbf{Pauli} $n$-\textbf{group} $\PN$ is the group whose elements have the form $\alpha(P_1,\ldots,P_n)$, where $(P_i)_{i=1}^n$ is an $n$-tuple of  Pauli operators,  $P_i\in\{X, Y, Z, I\}$, with global phase $\alpha \in \{ \pm1, \pm i\}$. The multiplication is defined by $\alpha(P_1,\ldots , P_n) \beta(Q_1,\ldots , Q_n) = \gamma(R_1,\ldots,R_n)$, where $P_iQ_i = \gamma_iR_i$, $\gamma = \alpha\beta(\prod_i \gamma_i)$.
 The unit is $(I)_{i=1}^n$. The group $\PN$ acts on the Hilbert space of $n$-qubits $H_n\coloneqq (\mathbb{C}^2)^{\otimes n}$ via the action
\begin{equation}\label{equ: action}
\alpha(P_i)_{i=1}^n~\action~\ket{\psi_1}\otimes\cdots\otimes\ket{\psi_n}\coloneqq \alpha P_1\ket{\psi_1}\otimes\cdots\otimes P_n\ket{\psi_n}
\end{equation}
Given a subgroup $S\leq \PN$, the \textbf{stabiliser} of $S$ is the linear subspace
\[
V_S\coloneqq  \{\ket{\psi}\in H_n\mid \forall P\in S. P\action\ket{\psi}=\ket{\psi}\}.
\]
The subgroups of $\PN$ that stabilise non-trivial subspaces must be commutative, and only contain elements with global phases $\pm1$. 
We call such subgroups \emph{stabiliser subgroups}.

\subsection{Stabiliser subgroups induce XOR theories}

We are interested in a quantum measurement scenario where $n$ parties share an $n$-qubit state and can each choose to perform a local Pauli operator $X$, $Y$, or $Z$ on their respective qubit.
The set of measurement labels is thus $\MX = \bigcup_{i=1}^n\{x_i,y_i,z_i\}$, and the contexts
are subsets of $\MX$ that contain at most one element for each index $i$,
as each party can only perform at most one measurement.
In other words, the maximal contexts are the sets $\{m_1, \ldots, m_n\}$ where
$m_i \in \{x_i,y_i,z_i\}$, and the contexts are subsets of these.
Note that this is a Bell-type scenario, of the kind considered in discussions of non-locality, a particular case of contextuality.
However, we will often only need to consider a subset of the contexts -- and even a subset of the measurement labels -- in order to derive an AvN argument.

Let $P=\alpha (P_i)_{i=1}^n\in\PN$ be an element of the Pauli $n$-group and
$\kp\in H_n$ a state stabilised by $P$ (hence we must have $\alpha = \pm 1$). Then,
\[\alpha(P_1\otimes \cdots \otimes P_n)\kp = \kp\]
and so $\kp$ is a $\alpha$-eigenvector of $P_1\otimes \cdots \otimes P_n$.
Consequently, the expected value satisfies
\[\braket{\psi | P_1\otimes \cdots \otimes P_n}{\psi} = \alpha.\]
From footnote \ref{note1},
this means that, given $n$ qubits prepared on the state $\kp$,
any joint outcome $s \in \mathbb{Z}_2^n$ resulting from measuring $P_i$ at each qubit $i$ must 
have even (resp. odd) parity when $\alpha = 1$ (resp. $\alpha=-1$).

Therefore, to any stabiliser group $S\leq\PN$ we associate an XOR theory 
\[\XORTe{S}\coloneqq \{\varphi_P\mid P\in S\},\]
where
\[
\varphi_P \coloneqq \left(\bigoplus_{\substack{i\in\{1,\ldots,n\}\\ P_i \neq I}}\bar{P}_i \;=\; a\right)
\]
where $a\in\mathbb{Z}_2$ is such that the element $P$ has global phase $\alpha = (-1)^a$.
We say that $S$ is \textbf{AvN} if $\XORTe{S}$ is inconsistent.

Note that, from the discussion above, it follows that this is (a subset of) the XOR theory $\XORTe{e}$ of every empirical model $e$ obtained from local Pauli measurements on any $n$-qubit state stabilised by $S$, i.e. any $\kp \in V_S$.


Consequently, given an AvN subgroup $S\leq\PN$ and any state $\ket{\psi}\in V_S$,
the $n$-partite empirical model realised by $\ket{\psi}$ under the Pauli measurements is strongly contextual.
Indeed, the inconsistency of $\XORTe{S}$ implies the impossibility of finding a global assignment compatible with the support of the empirical model \cite{Abramsky2}:

\begin{proposition}\label{prop: paradox}
\emph{AvN} subgroups of $\PN$ give rise to strongly contextual empirical models admitting All-vs-Nothing arguments.
\end{proposition} 

\subsection{Galois connections and relations with logic}\label{sec: Galois}
Proposition \ref{prop: paradox} shows that an AvN argument is essentially a logical paradox: to a strongly contextual model corresponds an inconsistent XOR theory. Previous work has already outlined this intriguing connection between contextuality and logical paradoxes. In \cite{Abramsky4} it is shown how quantum violations to  Bell  inequalities can be systematically viewed as arising from logical inconsistencies. In \cite{Abramsky2}, a structural equivalence between PR-box models and liar paradoxes is observed.
Here, the link between logical XOR paradoxes and the contextuality of stabiliser subgroups is formalised using a Galois connection.

There is a well-known correspondence between the rank of a subgroup $S$ of $\PN$ and the dimension of its stabiliser \cite{Caves}:
\begin{equation}\label{equ: rank}
\rank S=k\biff \dim V_S=2^{n-k}.
\end{equation}
In particular, when $S$ is a maximal stabiliser subgroup, we have $k=n$, hence $\dim V_S= 1$, so $S$ stabilises a unique state (up to global phase). We call such states \textbf{stabiliser states}.

We formalise and extend this link by introducing a Galois correspondence between subgroups of $\PN$ and their stabilisers.

Given two partially ordered sets $A$ and $B$, an \textbf{(antitone) Galois connection} between $A$ and $B$ is a pair $\langle f, g\rangle$ of order-reversing maps $f:A\rightarrow B, g:B\rightarrow A$ such that $a\leq gf(a)$ for all $a\in A$, and $b\leq fg(b)$ for all $b\in B$. 

Define a relation $R\subseteq \PN\times H_n$ by 
\[
gRv\biff g\action v=v.
\]
This induces a Galois connection between the powersets $\mathcal{P}(\PN)$ and $\mathcal{P}(H_n)$, ordered by inclusion:
\begin{equation}\label{equ: Galois relation}
\begin{matrix}
\mathcal{P}(\PN) & \longleftrightarrow & \mathcal{P}(H_n)\\
S & \longmapsto & S^\perp\coloneqq \{v\mid\forall g\in S. gRv\}\\
V^\perp\coloneqq \{g\mid\forall v\in V. gRv\} & \longmapsfrom & V.
\end{matrix}
\end{equation}
Note that \textbf{closed sets} $S^{\perp\perp}$ and $V^{\perp\perp}$ are subgroups of $\PN$
and vector subspaces of $H_n$, respectively. Therefore, by restricting \eqref{equ: Galois relation} to closed sets, we obtain a Galois connection $\SG(\PN)\leftrightarrow\SSet(H_n)$ between the closed subgroups of $\PN$ and the closed subspaces of $H_n$. Explicitly, the connection is given by the maps $F:\SG(\PN)\leftrightarrow \SSet(H_n):G$, with 
\[
F:: S\mapsto V_S \quad\quad \text{and} \quad\quad  G:: V  \mapsto  \bigcap_{\ket{\psi}\in V}(\PN)_{\ket{\psi}}, 
\]
where $(\PN)_{\ket{\psi}}\coloneqq \{A\in \PN\mid A\action \ket{\psi}=|\psi \rangle\}$ denotes the isotropy group of $\ket{\psi}$. Since the Galois connection is constructed through closed sets, it is \textbf{tight} in the sense of \cite{Raney}, and hence $F$ is an order-isomorphism, with inverse $G$. 


This result suggests an intriguing relation with the Galois connection between syntax and semantics in logic~\cite{Lawvere, Smith}
\begin{equation}\label{equ: logic}
\begin{matrix}
\text{$\mathcal{L}$-\textbf{Theories}} & \longleftrightarrow & \mathcal{P}(\text{$\mathcal{L}$-\textbf{Structures}})\\
\Gamma & \longmapsto & \{\mathcal{M}\mid\forall\varphi \in\Gamma. \mathcal{M}\models \varphi\}\\
\{\varphi\mid\forall\mathcal{M} \in M. \mathcal{M}\models\varphi\} & \longmapsfrom & M,

\end{matrix}
\end{equation}
where $\mathcal{L}$ is a formal language. The Galois closure operators correspond to deductive closure on the side of theories, and closure under the equivalence induced by the logic on the side of models.
Let $\Theories$ denote the set of all XOR theories. The map
\[
\XORT: \SG(\PN)\longrightarrow \Theories:: S\longmapsto \XORTe{S},
\]
is order-preserving, and allows us to establish a link between the Galois connection $\SG(\PN)\leftrightarrow \SSet(H_n)$ and the one described in \eqref{equ: logic}. In particular, we have the following commutative diagram

\begin{center}
\begin{tikzpicture}[auto]
\matrix (m) [matrix of math nodes, row sep=4em, column sep=4em, text height=1.5ex, text depth=0.25ex]
{
\SG(\PN) & \SSet(H_n)\\
\phantom{~~}\Theories\phantom{~~} & \mathcal{P}(\Structures)\\
};

\path[->, font=\scriptsize]
(m-1-1) edge[auto, bend left=10] node {$F$} (m-1-2)
edge[swap] node {$\XORT$} (m-2-1)
(m-1-2) edge node[auto] {$\XORM$} (m-2-2)
edge[auto, bend left=10] node {$G$} (m-1-1)
(m-2-1) edge[auto, bend left=10] node {} (m-2-2)
edge[auto, bend left=10] node {} (m-2-2)
(m-2-2) edge[auto, bend left=10] node {} (m-2-1);

\end{tikzpicture}
\end{center}
where $\Structures$ denotes the set of XOR-structures, and the order preserving function $\XORM$ maps a subspace $V$ of $H_n$ to the set 
\[
\XORMt{V}\coloneqq \left\{\mathcal{M}\mathrel{\Big|}\forall\varphi \in\XORTe{G(V)}. \mathcal{M}\models\varphi\right\}.
\]
%
Note that, in categorical terms, an antitone Galois connection is a dual adjunction between poset categories. One can easily show that the pair of functors $\tuple{\XORT, \XORM}$ is a monomorphism of adjunctions \cite{MacLane}.

\section{Characterising AvN arguments}\label{sec: AvN triple conjecture}

Using the general theory of AvN arguments for stabiliser states reviewed in the last section, we present a characterisation of AvN arguments based on the combinatorial concept of an AvN triple \cite{Abramsky2}. To achieve this result we will take advantage of the graph state formalism, which will also enable us to conclude that a tripartite GHZ state always underlies any AvN argument for stabiliser states. 

\subsection{AvN triples}\label{sec: AvN defn}

Since AvN subgroups give rise to strongly contextual empirical models, we are naturally interested in characterising this property. In \cite{Abramsky2}, this problem is addressed by introducing the notion of AvN triple. We rephrase the definition from \cite{Abramsky2} in slightly more general terms:

\begin{defn}[cf. {\cite[Definition 3]{Abramsky2}}]\label{defn: AvNtriple}
An \textbf{AvN triple} in $\PN$ is a triple $\tuple{e,f,g}$ of elements of $\PN$ with global phases $\pm 1$ that pairwise commute and that satisfy the following conditions:
\begin{enumerate}
\item\label{cond1} For each $i=1,\dots, n$, at least two of $e_i, f_i, g_i$ are equal.
\item\label{cond2} The number of $i$ such that $e_i=g_i\neq f_i$, all distinct from $I$, is odd. 
\end{enumerate}
\end{defn}

Note that the only difference with respect to the original definition from \cite{Abramsky2} is that we allow elements of an AvN triple to have global phase $-1$.

A key result from \cite{Abramsky2} is that AvN triples provide a sufficient condition for All-vs-Nothing proofs of strong contextuality. A similar argument shows that this is still true for the slightly more general notion of AvN triple. 

\begin{theorem}[cf. {\cite[Theorem 4]{Abramsky2}}]\label{thm:AvNtriple->Avn}
Any subgroup $S$ of $\PN$ containing an AvN triple is AvN.
\end{theorem}
\begin{proof}
Let $\tuple{e,f,g}$ be an AvN triple in $\PN$, where $e=(-1)^a (e_i)_{i=1}^n$, $f=(-1)^b(f_i)_{i=1}^n$ and $g=(-1)^c (g_i)_{i=1}^n$, with $a,b,c \in\mathbb{Z}_2$.  
We denote by $N_f$ the number of $i$'s such that $e_i=g_i\neq f_i$ and $e_i,f_i,g_i\neq I$, which is odd by \ref{cond2}. The global phase of $h\coloneqq efg$ is given by 
\[
(-1)^{a+b+c+N_f}=(-1)^{a\oplus b\oplus c\oplus 1}.
\]
Thus, if we consider the four equations corresponding to these elements in the XOR theory of the subgroup, summing their right-hand sides yields $a\oplus b\oplus c\oplus (a\oplus b\oplus c \oplus 1)=1$. On the other hand, by \ref{cond1}, we have $\{e_i,f_i,g_i\}=\{P,Q\}$ with at least two elements equal to $P$. Thus, by \eqref{equ: Pauli's}, the product $e_if_ig_i$ will be $Q$ up to a phase, and so, as this is absorbed into the global phase, $h_i=Q$. This means that each column of the four equations contains $2$ $P$s and $2$ $Q$s. Therefore, summing all the four equations we obtain $0=1$.  
\end{proof}

Remarkably, every AvN argument which has appeared in the literature can be seen to come down to exhibiting a AvN triple \cite{Abramsky2}. It was conjectured by the first author that the presence of an AvN triple in a stabiliser subgroup is also necessary for the existence of an All-vs-Nothing proof of strong contextuality. This  is the \textbf{AvN triple conjecture} \cite{AvN}. We will now prove the AvN triple conjecture for the case of maximal stabiliser subgroups, or equivalently, for stabiliser states, by taking advantage of the graph state formalism, which is briefly reviewed in the following subsection.

\subsection{Graph states}
Graph states are special types of multi-qubit states that can be represented by a graph.
Let $G=(V,E)$ be an undirected graph.
For each $u \in V$, consider the element $g^{u} = (g^{u}_v)_{v \in V} \in \Pauli_{|V|}$ with global phase $+1$ and
components
\begin{equation}\label{equ: graph state comps}
g^{(u)}_v =
\begin{cases}
X & \text{if $v = u$}\\
Z & \text{if $v \in \mathcal{N}(u)$}\\
I & \text{otherwise,}
\end{cases} 
\end{equation}
where $\mathcal{N}(u)$ denotes the neighborhood of $u$,
i.e. the set of all the vertices adjacent to $u$.
The \textbf{graph state} $\ket{G}$ associated to $G$ is the unique\footnote{
Uniqueness follows from \eqref{equ: rank} since there is an independent generator for each vertex.}
state stabilised by the subgroup generated by these elements,
\begin{equation}\label{equ: graph state}
S_G = \left\langle \{g^{(u)} \mid u \in V\}\right\rangle.
\end{equation}

One of the key properties of graph states is their generality with respect to stabiliser states, as stated in the following theorem due to Schlingemann \cite{Schlingemann}.
Recall the definition of the local Clifford (LC) group on $n$ qubits
\[
\Clifford_1^n\coloneqq \{U\in\textbf{U}(2)^n\mid U\Pauli_nU^\dagger=\Pauli_n\},
\]
where $U$ acts by conjugation on elements of $\PN$ via the representation of $\alpha(P_i)_{i=1}^n$ as the operator $\alpha P_1\otimes\cdots\otimes P_n$. Two states $\kp,\ket{\psi'}\in H_n$ are said to be LC-equivalent whenever there is a $U\in\Clifford_1^n$ such that $\ket{\psi'}=U\kp$.
\begin{theorem}[{\cite{Schlingemann}}]\label{thm: LU}
Any stabiliser state $\ket{S}$ is LC-equivalent to some graph state $\ket{G}$, i.e. $\ket{S}=U\ket{G}$ for some LC unitary $U\in\Clifford_1^{|V|}$. 
\end{theorem}

An instance of this result  will be important for us. 
Consider the $n$-partite GHZ state
\[
\GHZn=\frac{1}{\sqrt{2}}(\ket{0}^{\otimes n}+\ket{1}^{\otimes n}).
\]
We apply a local Clifford transformation consisting of a Hadamard unitary
\[
H\coloneqq \frac{1}{\sqrt{2}}\left(
\begin{matrix}
1 & 1\\
1 & -1
\end{matrix}
\right)\]
at every qubit of GHZ except the $k$-th one (where $1\leq k\leq n$ can be chosen arbitrarily)
to obtain
\[
\ket{\psi}=\frac{1}{\sqrt{2}}\left(\ket{+\cdots +0_k+\cdots +}\quad+\quad \ket{-\cdots -1_k-\cdots -}\right).
\]
%
%
%
%
%
The stabiliser group of $\ket{\psi}$ is generated by the elements $e, f^1,f^2,\dots, f^{k-1},f^{k+1},\dots, f^n$, where 
\[
e_i=
\begin{cases}
X & \text{if $i = k$}\\
Z & \text{if $i \neq k$}
\end{cases}
\quad \text{and}\quad 
f^j_i=
\begin{cases}
X & \text{if $i = j$}\\
Z & \text{if $i = k$}\\
I & \text{otherwise}
\end{cases}
\]
%
%
%
%
%
%
%
%
%
%
%
%
%
%
%
%
%
By the definition of a graph state, this list of stabilisers coincides with the \emph{star graph} centered at the vertex corresponding to the $k$-th qubit (Figure~\ref{fig: star}). Since $k$ was chosen arbitrarily, all the graph states corresponding to star graphs on $n$ vertices and different centers are LC-equivalent. 
\begin{figure}[htbp]
\centering
\includegraphics[scale=0.65]{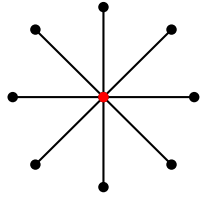}
\caption{Example of a star graph $G=(V,E)$ with $|V|=9$. The graph state $\ket{G}$ is LU-equivalent to the $9$-partite GHZ state.}\label{fig: star}
\end{figure}

Another important property of graph states is that they allow us to characterise LC-equivalence between them by a simple operation on the underlying graphs. This is the notion of local complementation. Given a graph $G=(V,E)$ and a vertex $v\in V$, the \textbf{local complement of $G$ at $v$}, denoted by $G \star v$, is obtained by complementing the subgraph of $G$ induced by the neighborhood $\mathcal{N}(v)$ of $v$ and leaving the rest of the graph unchanged \cite{Bouchet1993}. The following theorem is due to Van den Nest, Dehaene, and De Moor \cite{VandenNest} (see also \cite{Hein2004}).

\begin{theorem}[{\cite[Theorem 3]{VandenNest}}]\label{thm: local complementation}
By local complementation of a graph $G=(V,E)$ at some vertex $v\in V$ one obtains an LC-equivalent graph state.
Moreover, two graph states $\ket{G}$ and $\ket{G'}$ are LC-equivalent if and only if the corresponding graphs are related by a sequence of local complementations, i.e.  
\[
G'= G \star v_1 \cdots \star v_n
\]
for some $v_1,\dots v_n\in V$. 
\end{theorem}

Thanks to this theorem we can show that the $n$-partite GHZ state is LC-equivalent both to the state corresponding to the star graph (as shown above), and the one corresponding to the complete graph on $n$ vertices.
Indeed, it is sufficient to choose a vertex $v$ of the complete graph and apply a local complementation to it to obtain a star graph centered at $v$, as illustrated in Figure~\ref{fig: complete}.

\begin{figure}[htbp]
\centering
\includegraphics[scale=0.65]{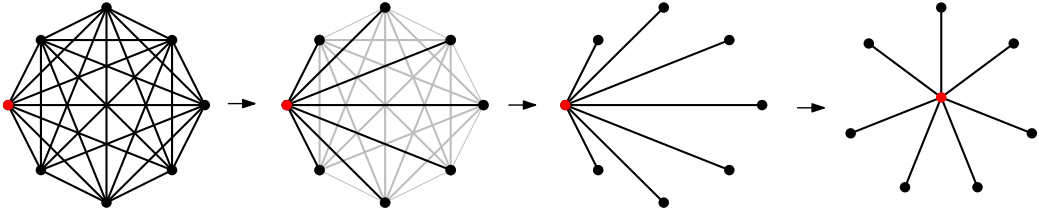}
\caption{The star graph centered at a vertex $v$ (marked in red) can be obtained from the complete graph by local complementation at the vertex $v$.}\label{fig: complete}
\end{figure}

\subsection{The AvN triple theorem and its consequences}
In this section, we prove the theorem characterising AvN arguments on stabiliser states.

Firstly, we need to make some observations.
Note that the Born rule is invariant under any unitary action acting simultaneously on the measurement by conjugation and on the state.
Therefore,
if we have a quantum realisable empirical model specified by a state $\kp$ and a set of measurements $\MX$,
then given any unitary $U$, the empirical model specified by the state $U\kp$ and the set of measurements $U \MX U^\dagger = \{ U A U^\dagger \mid A \in \MX\}$ is equivalent to the original one, in the sense that it assigns the same probabilities,
which of course implies that it has the same contextuality properties.

In the particular case when $\kp$ is a stabiliser state for the subgroup $S \leq \PN$,
and $U$ is a LC-operation, then 
the state $U \kp$ is a stabiliser state for the subgroup
$USU^\dagger=\{UPU^\dagger\mid P \in S\}$.

An important fact we shall need is that AvN triples are sent to AvN triples by such LC operations.
The reason is that LC operations are composed of local unitaries that act as  permutations on the set $\{\pm X, \pm Y, \pm Z\}$, and therefore preserve all the conditions of Definition~\ref{defn: AvNtriple}.

We now show that AvN triples fully characterise All-vs-Nothing arguments for stabiliser states, and that a tripartite GHZ state is always responsible for the existence of such an AvN proof of strong contextuality. 

\begin{theorem}[AvN Triple Theorem]\label{thm: AvN triple conjecture}
A maximal stabiliser subgroup $S$ of $\PN$ is AvN if and only if it contains an AvN triple. The AvN argument can be reduced to one concerning only three qubits. The state induced by the subgraph for these three qubits is LC-equivalent to a tripartite GHZ state. 
\end{theorem}
\begin{proof}
Sufficiency follows from Proposition~\ref{thm:AvNtriple->Avn}.
So, suppose that the maximal stabiliser subgroup $S$ is AvN.
Let $\ket{\psi}$ be the stabiliser state corresponding to $S$.
Since any stabiliser state is LC-equivalent to a graph state by 
Theorem \ref{thm: LU}, and since $LC$ transformations preserve AvN triples,
we can suppose without loss of generality that $\ket{\psi}$ is a graph state $\ket{G}$ induced by a graph $G=(V,E)$, and consequently that $S = S_G$ as in \eqref{equ: graph state}.
.

Given that the empirical model obtained from the state $\ket{G}$ and local Pauli operators is strongly contextual, there must exist at least one vertex $u$ with degree at least $2$, \ie $|\mathcal{N}(u)| \geq 2$.
Indeed, if $G$ has no such vertex, $G$ is a union of disconnected edges and vertices, which implies that $\ket{G}$ is a tensor product of $1$-qubit and $2$-qubit states, which do not present strongly contextual behavior for any choice of local measurements \cite{BrassardMethotTapp2005:MinimumEntangledStatePseudoTelepathy,ABCDKM}. 

Let $u\in V$ have degree $\ge 2$ and let $v,w$ be two distinct vertices in $\mathcal{N}(u)$. We have two possible cases:
\begin{enumerate}
\item There is an edge between $v$ and $w$.
Then, in accordance with \eqref{equ: graph state comps}, the elements $g^u,g^v,g^w$ of $S_G$
have the form:\footnote{The notation in \eqref{equ: AvN1} indicates that $g^u_u = X$, $g^u_v = g^u_w = Z$, and $g^u_z$ is either $Z$ or $I$ for every other vertex $z \in V \setminus \{u,v,w\}$, and analogously for the other lines.}
\begin{equation}\label{equ: AvN1}
\begin{array}{rcccl}
g^u : &  X_u & Z_v & Z_w & [\text{$I$ or $Z$ on all other qubits}]\\
g^v : &  Z_u & X_v & Z_w & [\text{$I$ or $Z$ on all other qubits}]\\
g^w : &  Z_u & Z_v & X_w & [\text{$I$ or $Z$ on all other qubits}],\\
\end{array}
\end{equation}
which are easily seen to constitute an AvN triple. 

\item There is no edge between $v$ and $w$. Then, we have
\[
\begin{array}{rcccl}
g^u : &  X_u & Z_v & Z_w & [\text{$I$ or $Z$ on all other qubits}]\\
g^v : &  Z_u & X_v & I_w & [\text{$I$ or $Z$ on all other qubits}]\\
g^w : &  Z_u & I_v & X_w & [\text{$I$ or $Z$ on all other qubits}]\\
\end{array}
\]
and the elements $\langle  g^u, g^ug^v, g^ug^w\rangle$ form an AvN triple:
\begin{equation*}\label{equ: AvN2}
\begin{array}{rcccl}
g^u    : &  X_u & Z_v & Z_w & [\text{$I$ or $Z$ on all other qubits}]\\
g^ug^v : &  Y_u & Y_v & Z_w & [\text{$I$ or $Z$ on all other qubits}]\\
g^ug^w : &  Y_u & Z_v & Y_w & [\text{$I$ or $Z$ on all other qubits}].\\
\end{array}
\end{equation*}
\end{enumerate}
Notice that in both cases the AvN argument is reduced to just three qubits. Moreover, by the discussion at the end of the previous subsection, we know that the state corresponding to the subgraph induced by $u,v,w$ in either of these two cases is LC-equivalent to a tripartite GHZ state: in Case 1 we have a complete graph on three vertices, while in Case 2 we have a star graph centered at $u$. 
\end{proof}

The second part of the above result means that the essence of the contradiction is witnessed by looking at only three qubits. In fact, in the contexts 
being considered, the experimenters at the remaining $n-3$ parties either perform no measurement or a $Z$ measurement.
We could imagine that, in trying to build a consistent global assignment of outcomes in $\mathbb{Z}_2$ to all the measurements,
each of these $n-3$ parties $i$ is allowed to  freely choose a value $0$ or $1$ for the variable $\bar{Z}_i$. Then, the equations for the variables representing the measurements of the remaining three parties would 
be those of the usual GHZ argument, up to flipping an even number of the values on the right-hand side. In terms of the state, we can use the  ``partial inner product'' operation described e.g.~in \cite[p.~129]{QPSI}\footnote{This is actually  the application of a linear map to a vector under Map-State duality \cite{abramsky2008categorical}.} to apply the eigenvectors corresponding to the chosen values for the other $n-3$ parties to $\ket{G}$, resulting in a three-qubit pure state which is LC-equivalent to the GHZ state.

From this theorem, we immediately obtain the following corollaries:
\begin{corollary}
A graph state $|G\rangle$ is strongly contextual if and only if $G$ has a vertex of degree at least $2$.
\end{corollary}

\begin{corollary}
Every strongly contextual 3-qubit stabiliser state is LC-equivalent to the GHZ state.
\end{corollary}

%
%
%

\begin{figure}[htbp]
\centering
\includegraphics[scale=0.6]{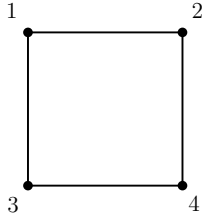}
\caption{The $4$-qubit $2$-dimensional cluster-state.}\label{fig: Cluster4}
\end{figure}

We provide some examples to clarify the statement of Theorem \ref{thm: AvN triple conjecture}.
Cluster states are a fundamental resource in measurement-based quantum computation \cite{Raussendorf, Raussendorf2, Raussendorf3}. The $4$-qubit $2$-dimensional cluster state is described by the graph in Figure \ref{fig: Cluster4}. 
Its stabiliser group $S$ is generated by the following elements of $\Pauli_4$:

\[
\begin{array}{rcccc}
g^1 : & X_1 & Z_2 & Z_3 & I_4\\
g^2 : & Z_1 & X_2 & I_3 & Z_4\\
g^3 : & Z_1 & I_2 & X_3 & Z_4\\
g^4 : & I_1 & Z_2 & Z_3 & X_4
\end{array}
\]
The stabiliser group $S$ contains the following 4 AvN triples, corresponding to the triples of qubits highlighted in Figure \ref{fig: Cluster4.1}:
\begin{equation}\label{eq:triples-eg1}
\begin{aligned}
\begin{array}{rcccc}
g^1    : & X_1 & Z_2 & Z_3 & I_4\\
g^1g^2 : & Y_1 & Y_2 & Z_3 & Z_4\\
g^1g^3 : & Y_1 & Z_2 & Y_3 & Z_4
\end{array}
& &
\begin{array}{rcccc}
g^2    : & Z_1 & X_2 & I_3 & Z_4\\
g^2g^1 : & Y_1 & Y_2 & Z_3 & Z_4\\
g^2g^4 : & Z_1 & Y_2 & Z_3 & Y_4
\end{array}
\\
\begin{array}{rcccc}
g^3    : & Z_1 & I_2 & X_3 & Z_4\\
g^3g^1 : & Y_1 & Z_2 & Y_3 & Z_4\\
g^3g^4 : & Z_1 & Z_2 & Y_3 & Y_4
\end{array}
& &
\begin{array}{rcccc}
g^4    : & I_1 & Z_2 & Z_3 & X_4\\
g^4g^2 : & Z_1 & Y_2 & Z_3 & Y_4\\
g^4g^3 : & Z_1 & Z_2 & Y_3 & Y_4
\end{array}
\end{aligned}
\end{equation}

\begin{figure}[htbp]
\centering
\includegraphics[scale=0.5]{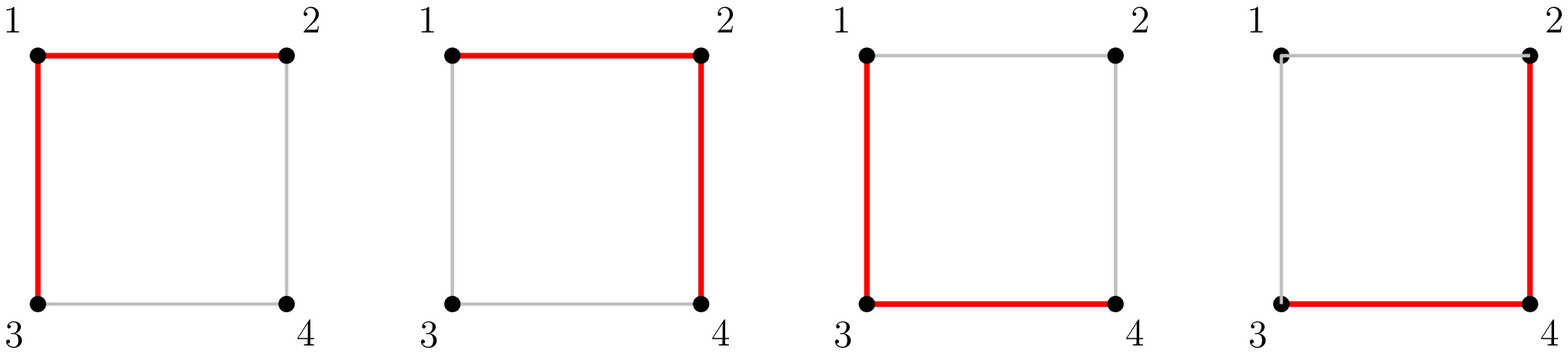}
\caption{Qubits generating the AvN triples of \eqref{eq:triples-eg1}. Each triple of qubits is LC-equivalent to GHZ.}\label{fig: Cluster4.1}
\end{figure}

\begin{figure}[htbp]
\centering
\includegraphics[scale=0.6]{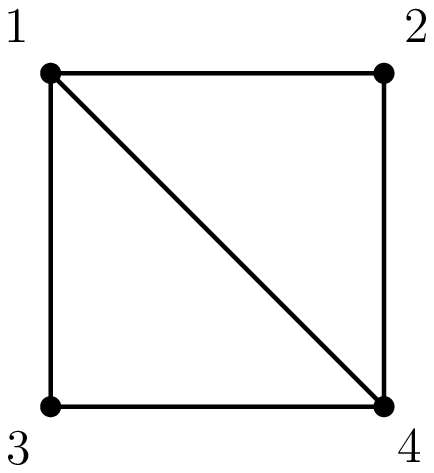}
\caption{The graph $G$ of the state $\ket{G}$.}\label{fig: graph}
\end{figure}

As another example, consider the graph state $\ket{G}$ represented in Figure \ref{fig: graph}.
Its stabiliser $S$ is generated by the following elements of $\Pauli_4$:
\[
\begin{array}{rcccc}
g_1 : & X_1 & Z_2 & Z_3 & Z_4\\
g_2 : & Z_1 & X_2 & I_3 & Z_4\\
g_3 : & Z_1 & I_2 & X_3 & Z_4\\
g_4 : & Z_1 & Z_2 & Z_3 & X_4
\end{array}
\]
and contains the following AvN triples: 
\begin{equation}\label{eq:triples-eg2}
\begin{aligned}
\begin{array}{rcccc}
g_1 : & X_1 & Z_2 & Z_3 & Z_4\\
g_3 : & Z_1 & I_2 & X_3 & Z_4\\
g_4 : & Z_1 & Z_2 & Z_3 & X_4\\
\end{array}
& &
\begin{array}{rcccc}
g_1 : & X_1 & Z_2 & Z_3 & Z_4\\
g_2 : & Z_1 & X_2 & I_3 & Z_4\\
g_4 : & Z_1 & Z_2 & Z_3 & X_4\\
\end{array}
\\
\begin{array}{rcccc}
g_1    : & X_1 & Z_2 & Z_3 & Z_4\\
g_1g_2 : & Y_1 & Y_2 & Z_3 & I_4\\
g_1g_3 : & Y_1 & Z_2 & Y_3 & I_4
\end{array}
& &
\begin{array}{rcccc}
g_4    : & Z_1 & Z_2 & Z_3 & X_4\\
g_4g_2 : & I_1 & Y_2 & Z_3 & Y_4\\
g_4g_3 : & I_1 & Z_2 & Y_3 & Y_4
\end{array}
\end{aligned}
\end{equation}
which correspond to the triples of qubits illustrated in Figure~\ref{fig: triangles}.

\begin{figure}[htbp]
\centering
\includegraphics[scale=0.6]{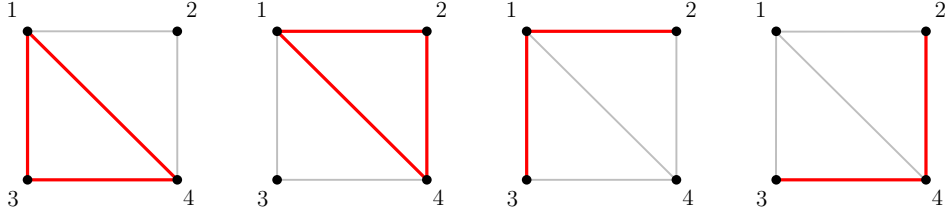}
\caption{Qubits generating the AvN triples of \eqref{eq:triples-eg2}. Each triple of qubits is LC-equivalent to GHZ.}\label{fig: triangles}
\end{figure}

%
%
%
%
%

\section{Applications}\label{sec: method}

In this section we take advantage of the characterisation introduced above to develop a computational method to identify all the possible AvN arguments.

\subsection{Counting AvN triples}
We start by introducing an alternative definition of AvN triple.

\begin{defn}[Alternative Definition of AvN triple]\label{defn: Alternative}
An \emph{AvN triple} in the Pauli $n$-group $\PN$ is a triple $\tuple{e,f,g}$ with global phases $\pm 1$, such that
\begin{enumerate}
\item\label{cond111} For each $i=1,\dots, n$, at least two of $e_i, f_i, g_i$ are equal.
\item\label{cond11} The number $N_g$ of $i$'s such that $e_i=f_i\neq g_i$, all distinct from $I$, is odd.
\item\label{cond22} The number $N_e$ of $i$'s such that $e_i\neq f_i= g_i$, all distinct from $I$, is odd.
\item\label{cond33} The number $N_f$ of $i$'s such that $e_i=g_i\neq f_i$, all distinct from $I$, is odd. 
\end{enumerate}
\end{defn}

The equivalence of the two definitions follows directly from the following lemma.

\begin{lemma}
Let $n\ge 3$. Suppose $e,f,g\in\PN$ have global phase $\pm 1$ and are such that for each $i=1,\dots, n$, at least two of $e_i, f_i, g_i$ are equal. Then $e,f,g$ commute pairwise if and only if $N_e, N_f$ and $N_g$ have the same parity.
\end{lemma}

\begin{proof}
Given two arbitrary elements $h,k\in\PN$ we have
\[
hk=(-1)^{|\{i\mid h_i\neq k_i \;\wedge\; h_i,k_i\neq I\}|}kh.
\]
Thus, $e$ and $f$ commute if and only if $N\coloneqq\left|\{i\mid e_i\neq f_i \;\wedge\; e_i,f_i\neq I\}\right|$ is even. By hypothesis, for each $i$, at least two of $e_i.f_i.g_i$ are equal, hence
\[
N=|\{i\mid g_i=e_i\neq f_i \;\wedge\; e_i,f_i,g_i\neq I\}|+|\{i \mid e_i\neq f_i=g_i \;\wedge\; e_i,f_i,g_i\neq I\}|=N_f+N_e.
\]
Therefore, 
\[
e,f \text{ commute } \biff N_e \text{ and } N_f \text{ have the same parity}.
\]
Similarly,
\[
\begin{split}
f,g \text{ commute } &\biff N_f \text{ and } N_g \text{ have the same parity}\\
e,g \text{ commute } &\biff N_e \text{ and } N_g \text{ have the same parity},
\end{split}
\]
and the result follows. 
\end{proof}

Note that this new definition can be used to derive an alternative proof of the fact that any AvN triple for $n$-partite states can be reduced to an AvN triple that only involves 3 qubits, in accordance with Theorem \ref{thm: AvN triple conjecture}. Indeed, given an AvN triple $\tuple{e,f,g}$ in $\PN$, since $N_g,N_e,N_f$ are odd, we can always choose 3 indices $1\leq i_1,i_2,i_3\leq n$ such that 
\[
e_{i_1} = f_{i_1}\neq g_{i_1}, \qquad e_{i_2} \neq f_{i_2} =g_{i_2},  \qquad e_{i_3} = g_{i_3}\neq f_{i_3}
\]
Clearly, the elements of the triple restricted to these indices constitute an AvN triple in $\Pauli_3$ and therefore an AvN argument. 

The rationale for introducing Definition \ref{defn: Alternative} is that it allows to better understand AvN triples from a computational perspective. We show a first example by providing a closed formula for the number of AvN triples in $\PN$.
\begin{proposition}\label{prop:countingAvN}
Let $n\ge 3$. The number of AvN triples in $\PN$ is given by
\[
8\sum_{k=1}^{\frac{1}{2}(n+[n])-1}{{n}\choose{2k+1}} {{k+1}\choose{k-1}}\cdot 6^{2k+1}\cdot 22^{n-2k-1},
\]
where $[n]\in\mathbb{Z}_2$ denotes the parity of $n$. 
\end{proposition}
\begin{proof}
The factor of $2^3=8$ corresponds to the possible choices of global phase $\pm 1$ for each element in the triple.

By Definition \ref{defn: Alternative}, an AvN triple $\tuple{e,f,g}$  is essentially determined by three odd numbers $N_e, N_f,N_g$. Their sum $S\coloneqq N_e+N_f+N_g\leq n$ can be seen as the number of columns of the triple that play an active part in the AvN argument. Let us compute the amount of AvN triples having $S$ ``relevant'' columns. We start by counting the number of solutions to the equation 
\[
N_e+N_f+N_g=S,
\]
where $N_e,N_f,N_g,S$ are all odd numbers. Let $k,e,f,g\ge 0$ be integers such that $S=2k+1$ and $N_i=2i+1$ for $i=e,f,g$. We have 
\[
N_e+N_f+N_g=S \biff 2e+1+2f+1+2g+1=2k+1 \biff e+f+g=k-1
\]
By stars and bars~\cite{Feller}, the number of solutions to this equation is ${{k+1}\choose{k-1}}$. By condition \ref{cond111} of Definition \ref{defn: Alternative} we must choose two observables in $\{X,Y,Z\}$ (the order counts) in each of the $S$ relevant columns, for a total of $P_{(3,2)}^S=6^{2k+1}$. Finally, we have $8$ possible configurations of each of the remaining $n-S$ non-relevant columns, namely
\[
\begin{matrix}
I\\
I \\
I
\end{matrix}
\qquad
\begin{matrix}
P\\
P \\
P
\end{matrix}
\qquad 
\begin{matrix}
P & I & I\\
I & P & I \\
I & I & P &
\end{matrix}
\qquad 
\begin{matrix}
P & P & I\\
P & I & P \\
I & P & P &
\end{matrix}
\]
where $P$ has to be chosen in $\{X,Y,Z\}$ for a total of $(3\cdot 7+1)^{n-S}=22^{n-2k-1}$ possibilities. Hence, the number of AvN triples in $\PN$ having $S=2k+1$ relevant columns is 
\[
N_S\coloneqq {{k+1}\choose{k-1}}\cdot 6^{2k+1}\cdot 22^{n-2k-1}.
\]
Now, the amount of odd numbers of relevant columns $S\leq n$ that we can select is given by 
\[
\sum_{k=1}^{\frac{1}{2}(n+[n])-1}{{n}\choose{2k+1}},
\]
and the result follows. 
\end{proof}

\subsection{Generating AvN triples}

We devote this last section to the presentation of a computational method to generate all the AvN triples contained in $\PN$. Until now, we only had a rather limited number of examples of quantum-realisable models featuring All-vs-Nothing proofs of strong contextuality. Thanks to the AvN triple theorem \ref{thm: AvN triple conjecture}, the technique we introduce allows us to find all such models for a sufficiently small $n$. 

\textbf{Check vectors} \cite{Chuang} are a useful way to represent elements of $\PN$ in a computation-friendly way. Given an element $P\coloneqq \alpha(P_i)_{i=1}^n\in\PN$, its check vector $r(P)$ is a $2n$-vector
\[
r(P)=(x_1,x_2,\dots, x_n, z_1, z_2,\dots, z_n)\in\mathbb{Z}_2^{2n}
\]
whose entries are defined as follows
\[
(x_i,z_i)=
\begin{cases}
(0,0) & \text{ if } P_i=I\\
(1,0) & \text{ if } P_i=X\\
(1,1) & \text{ if } P_i=Y\\
(0,1) & \text{ if } P_i=Z.
\end{cases}
\]
Every check vector $r(P)$ completely determines $P$ up to phase (i.e. $r(P)=r(\alpha P)$ for all $\alpha\in\{\pm 1,\pm i\}$). We can use this representation to express the conditions for an AvN triple. 
More specifically, we represent an AvN triple, up to the global phases of each of its elements,
as a matrix $M\in M_{3\times 2n}(\mathbb{Z}_2)$ whose rows are the check vectors of each element of the triple. Condition \ref{cond111} of Definition \ref{defn: Alternative} can be rewritten as 
\begin{equation}\label{equ: condition1}
\forall 1\leq j\leq n. \; \exists i,k\in\{1,2,3\}. \;
\begin{cases}
M_{i,j}=M_{k,j} & \\
M_{i,n+j}=M_{k,n+j}
\end{cases}
\end{equation}
The numbers $N_e,N_f, N_g$ can also be easily computed. For instance, $N_g$ equals the cardinality of the set 
\[
\begin{split}
\{i\in\{1,2,3\} \mid  M_{1,j}=M_{3,j} \;\wedge\; & M_{1,n+j}=M_{3,n+j} \;\wedge\; (M_{1,j}\neq M_{2,j} \vee M_{1,n+j}\neq M_{2,n+j})\\
\;\wedge\; & (M_{1,j}\neq 0 \vee M_{1,n+j}\neq 0) \;\wedge\; (M_{2,j}\neq 0 \vee M_{2,n+j}\neq 0)\;\}
\end{split}
\]
Hence, in order to find all the AvN triples in $\PN$ we need to solve the following problem:
\[
\begin{aligned}
& \textbf{Find all} & ~~~~~ & M \in M_{3\times 2n}(\mathbb{Z}_2) \\
& \textbf{such that} & ~~~~~ & M \text{ verifies } \eqref{equ: condition1},\\
& & ~~~~~ & \text{$N_g$, $N_e$, and $N_f$ are odd},\\
\end{aligned}
\]
which is easily programmable.

An implementation of this method using Mathematica \cite{Mathematica} can be found in \cite{Code}, where we present the algorithm and the resulting list of all $216$ AvN triples in $\Pauli_3$ and all $19008$ AvN triples in $\Pauli_4$, disregarding the choice of global phases $\pm1$ for each element --
in order to get the total number of AvN triples from Proposition~\ref{prop:countingAvN}, note that these numbers need to be multiplied by a factor of $8$ to account for this choice of these global phases.
By Theorem \ref{thm: AvN triple conjecture}, this list generates all the possible AvN arguments for $3$-qubit and $4$-qubit stabiliser states. 

\section{Conclusions}
The recent formalisation and generalisation of All-vs-Nothing arguments in stabiliser quantum mechanics \cite{Abramsky2} allowed us to study their properties from a purely mathematical standpoint.

Thanks to this framework, we have introduced an important characterisation of AvN arguments based on the combinatorial concept of AvN triple \cite{Abramsky2}, leading to a computational technique to identify all such arguments for stabiliser states. 
The graph state formalism, which played a crucial r\^{o}le in the proof of the AvN triple theorem, also allowed us to infer an important structural feature of AvN arguments, namely that any such argument can be reduced to an AvN proof on three qubits, 
which is essentially a standard GHZ argument.
This result shows in particular that the GHZ state is the only 3-qubit stabiliser state, up to LC-equivalence, admitting an AvN argument for strong contextuality.  

Our computations provide a very large number of quantum-realisable strongly contextual empirical models admitting AvN arguments. These new models could potentially find applications in quantum information and computation, as well as contributing to the ongoing theoretical study of strong contextuality as a key feature of quantum mechanics \cite{Abramsky1, Abramsky2, Abramsky3, ABCDKM}.

The abstract formulation of generalised AvN arguments has also allowed us to introduce new insights into the connections between logic and the study of contextuality. Recent work on logical Bell inequalities \cite{Abramsky4} and the relation between contextuality and semantic paradoxes \cite{Abramsky2} suggests a strong connection between these two domains. In this work, we have taken a first step towards a formal characterisation of this link in the quantum-realizable case by showing the existence of a Galois connection between subgroups of the Pauli $n$-group and subspaces of the Hilbert space of $n$-qubits, which can be seen as the stabiliser-theoretic counterpart of the Galois connection between syntax and semantics in logic.

\ack{
The authors would like to thank Nadish de Silva, Kohei Kishida, and Shane Mansfield for helpful discussions. This work was carried out in part while the authors visited the Simons Institute for the Theory of Computing (supported by the Simons Foundation) at the University of California, Berkeley, as participants of the Logical Structures in Computation programme.}

\funding{Support from the following is gratefully acknowledged: 
the Simons Institute for the Theory of Computing;
EPSRC EP/N018745/1, `Contextuality as a Resource in Quantum Computation' (SA, RSB);
EPSRC Doctoral Training Partnership and Oxford--Google Deepmind Graduate Scholarship (GC).}

\bibliographystyle{vancouver}
\bibliography{AvN}

%
%
%
%
%

\end{document}